\documentclass[pra,superscriptaddress,twocolumn,longbibliography,reprint]{revtex4-2}
\usepackage{dcolumn}
\usepackage[T1]{fontenc}
\usepackage{amsmath,amsfonts,amssymb,amsthm,graphics,graphicx,epsfig,color,times,natbib,bm,makecell,soul}
\usepackage{booktabs}
\usepackage{dsfont}
\usepackage{subfigure}
\usepackage{hyperref}
\usepackage{cleveref}
\usepackage{appendix}
\usepackage[ruled]{algorithm2e}
\usepackage{float}
\usepackage{ulem}

\hypersetup{colorlinks=true, citecolor=blue, urlcolor=blue, linkcolor=blue} 

\theoremstyle{definition}
\newtheorem*{lemma}{Lemma}
\newtheorem{theorem}{Theorem} 
\newtheorem*{example}{Example}
\allowdisplaybreaks[4]
\soulregister\cite7

\begin{document}
	\date{\today}
	\title{Tight upper bound and monogamy relation for the maximum quantum value of the parity-CHSH inequality and applied to device-independent randomness}
	
	\author{Guannan Zhang}
	\affiliation{College of Science, China University of Petroleum, 266580 Qingdao,  China.}
	
	\author{Jiamin Xu}\email{daniel@jiuzhangqt.com} 
	\affiliation{Jiuzhang Quantum Technology Co., Ltd., 250102 Jinan, P.R. China.}
	
	\author{Ming Li}\email{liming@upc.edu.cn.}
	\affiliation{College of Science, China University of Petroleum, 266580 Qingdao,  China.}
	
	\author{Shuqian Shen}
	\affiliation{College of Science, China University of Petroleum, 266580 Qingdao,  China.}
	
	\author{Lei Li}
	\affiliation{College of Science, China University of Petroleum, 266580 Qingdao,  China.}
	
	\author{Shaoming Fei}
	\affiliation{School of Mathematical Sciences, Capital Normal University, 100048 Beijing, China}
	
	\begin{abstract}
		Based on the violation of Bell inequalities, we can verify quantum random numbers by examining the correlation between device inputs and outputs. 
		In this paper, we derive the maximum quantum value of the parity-CHSH inequality for a three-qubit system, establishing a tight upper bound applicable to any quantum state. Simultaneously, the necessary constraints for achieving saturation are analyzed. Utilizing this method, we present necessary and sufficient conditions for certain states to violate the parity-CHSH inequality. Building upon our proposal, the relationship between the noise parameter and the certifiable randomness in a bipartite entangled state is probed. Furthermore, we derive a monogamy relationship between the average values of the parity-CHSH inequality associated with the reduced three-qubit density matrices of GHZ-class states comprising four qubits.
	\end{abstract}
	\maketitle

	\section{INTRODUCTION}
	Bell inequality is a crucial topic in the field of quantum information, first proposed to address the EPR paradox proposed by Einstein, Podolsky and Rosen \cite{1}. For a given local hidden variable (LHV) model \cite{2}, the Bell inequality can be represented geometrically as a polyhedron, where the faces of the polyhedron correspond to the various Bell inequalities. So, all Bell inequalities describe the corresponding LHV models. The incompatibility between the correlations predicted by quantum mechanics and those allowed by the LHV model is known as nonlocality \cite{3}, which can be evidenced by the violation of the Bell inequality. As a non-classical phenomenon, Bell nonlocality plays a major role in quantum applications such as quantum randomness \cite{4,5}, quantum key distribution \cite{quantumkeydistribution1,quantumkeydistribution2}.	Monogamy is an important feature of nonlocality, indicating that in a quantum system, strong nonlocality restricts the nonlocal correlations that other participants can have with one of the parties. Such limitations reflect the unique nature of quantum entanglement \cite{monogamy1,monogamy2} and are essential for ensuring security in quantum communication \cite{monogamy3}, as well as for understanding the nature of quantum correlations in multipartite systems.
	
	Compared to the Bell inequality of  two-qubit system, quantum nonlocality in the case of multipartite exhibits a more complex structure. Based on the parity-CHSH game, Ribeiro has proposed a new multipartite Bell inequalities \cite{Parity-CHSH}, which extends the CHSH inequality to N parties.
	The parity-CHSH inequality is suitable for device-independent (DI) quantum key distribution and must satisfy structural requirements that prohibit the simultaneous achievement of perfectly correlated measurement results among all parties, as well as ensuring a sufficiently high violation of the Bell inequality and imposing additional constraints on the Bell settings. 
	
	Randomness is a crucial research issue in the field of information security, especially in cryptography. Almost all cryptographic algorithms and protocols require data that must be secret to attackers, and these keys must be random numbers, such as DI conference key agreement (DICKA) \cite{DICKA1,DICKA2} and DI randomness expansion (DIRE) \cite{DIRE1,DIRE2,DIRE3}. Unlike pseudorandom numbers generated based on classical mechanics 'surface randomness', in DI random number extensions, due to the intrinsic randomness within quantum systems, true random numbers independent of external variables can be generated utilizing the nonlocality of Bell inequality. In the Bell scenario,  min-entropy \cite{minimumentropy} is an effective tool for quantifying the randomness contained in the output sequence by analyzing the correlation between input and output classical bits. By combining the NPA hierarchy \cite{NPA1,NPA2,NPA3,NPA4} and the semidefinite programming (SDP) algorithm, the lower bound of minimum entropy can be obtained.
	
	For a given Bell inequality, calculating the maximum quantum value of a specific quantum state is essential for studying randomness and other related aspects. In this paper, we derive the tight upper bound of the maximum quantum value of the parity-CHSH inequality for any quantum state in three-qubit system, and we provide the corresponding constraint conditions that the quantum state must satisfy. Based on this method, we present the maximum quantum value of a state with mixed white noise and analyze the lower bound of the maximum probability distribution under the condition of reaching the maximum quantum violation. Moreover, we compare it with the upper bound of the numerical solution of the maximum probability distribution obtained through SDP,finding that randomness can gain an accurate numerical estimate. We also discuss the monogamy relation between the maximum quantum values for 4-qubit GHZ-class states, which reflects the distribution of quantum nonlocality in multipartite systems.
	
	\section{ the maximum violation for parity-CHSH inequality}
	Consider a Bell scenario with three parties, Alice and Bob performing two measurements labeled by $x,y \in {\left\lbrace 0,1\right\rbrace }$, with binary outcomes $a,b \in {\left\lbrace -1,1\right\rbrace }$ respectively. And Charlie perform with just one measurement with binary outcomes, donating $z=0,c\in {\left\lbrace -1,1\right\rbrace }$. We briefly explore the parity-CHSH inequality in three-qubit systems. Let $A_{x},B_{y},C_{0}$ be observables of the form $G=\vec{g}\cdot\vec{\sigma}=\sum_{i=1}^{3}g_{i}\sigma_{i}$, with $\vec{g}=\left( g_{1},g_{2},g_{3} \right) \in \left \{ \vec{a},\vec{a^{'}},\vec{b},\vec{b^{'}},\vec{c} \right \}$ be the three-dimensional real unit vectors, corresponding to $G\in \left \{ A_{0},A_{1},B_{0},B_{1},C_{0} \right \}$, and $\vec{\sigma}=(\sigma_{1},\sigma_{2},\sigma_{3})$, $\sigma_{i}(i=1,2,3)$ are the Pauli matrices. The parity-CHSH inequality is defined as follows:
	\begin{align}\label{eq1}
		\mathcal{B}
		&=\left \langle A_{1}B_{-}C_{0} \right \rangle +\left \langle A_{0}B_{+}  \right \rangle 
		\leqslant_{L} 1 \leqslant_{Q} \sqrt{2},
	\end{align}
	where $B_{\pm}=\frac{1}{2}(B_{0}\pm B_{1})$. The classical bound for LHV models is 1 and the maximum quantum value is $\sqrt{2}$ given by Ref. \cite{Parity-CHSH, boost}. Let $\rho$ be a general three-qubit state,
	\begin{gather}\label{eq2}
		\rho = \frac{1}{8}\sum_{i,j,k=0}^{3}t_{ijk}\sigma_{i}\sigma_{j}\sigma_{k},
	\end{gather}
	where $\sigma_{0}=I_{0}$ is the 2$\times$2 identity matrix, and
	\begin{gather*}
		t_{ijk} ={\rm Tr}(\rho(\sigma_{i}\otimes \sigma_{j}\otimes \sigma_{k})).
	\end{gather*}
	
	Prior to proving Theorem $\ref{thm1}$ on the violation of the parity-CHSH inequality, we present the following lemma from Ref. \cite{lemma}.
	\begin{lemma}\label{lemma}
		Let G be a rectangular matrix of size n$\times$n. For any vectors $\vec{x},\vec{y}\in R^{n} $ , we have that
		\begin{gather*}
			\left | \vec{x}^{T}G\vec{y} \right | \leqslant \lambda _{\max}\left | \vec{x} \right |\left | \vec{y} \right |  ,
		\end{gather*}
		where $\lambda_{\max}$ is the largest singular value of the matrix G. When $\vec{x}$ and $\vec{y}$ are singular vectors corresponding to the maximum singular value $\lambda_{\max}$ of matrix G, the inequality can be taken as equal.
	\end{lemma}
	\begin{theorem}
		\label{thm1}
		For any three-qubit quantum state $\rho$, the maximum quantum value $\mathcal{B}_{Q}$ of the parity-CHSH inequality satisfies
		\begin{align}\label{eq3}
			\mathcal{B}_{Q}
			\equiv\underset{A_{x},B_{y},C_{0}}{\max} \left | \mathcal{B}^{(3)}_{\rho } \right | \leqslant\underset{\vec{c}}{\max}\sqrt{\lambda_{1}^{2}+\lambda_{2}^{2}},
		\end{align}
		where $\mathcal{B}_{\rho }={\rm Tr}(\mathcal{B}\rho)$, and $\lambda_{1}$, $\lambda_{2}$ represent the maximum singular value of the matrices $\sum_{k=1}^{3}c_{k}T_{k}$ and $T_{0}$, respectively. Here $\sum_{k=1}^{3}c_{k}T_{k}$ has entries $(T_{k})_{ij}=t_{ijk}$, with $c_{k}$ being the entry of unit vector $\vec{c}$ , while $T_{0}$ has entries defined as $(T_{0})_{ij}=t_{ij0}$. 
	\end{theorem}	
	\begin{proof}[proof]
		One computes that the mean value of ABC and AB to be,
		\begin{align}\label{eq4}
			\left \langle ABC\right \rangle
			&=\frac{1}{8}\sum_{i,j,k=1}^{3}t_{ijk}a_{i}b_{j}c_{k}{\rm Tr}(\sigma _{i}^{2}\otimes\sigma _{j}^{2}\otimes\sigma _{k}^{2}) \nonumber\\
			&=\sum_{i,jk=1}^{3}t_{ijk}a_{i}b_{j}c_{k} 
			=\left \langle \vec{a},\sum_{k=1}^{3}c_{k}T_{k}\vec{b} \right \rangle  \nonumber \\
			\left \langle AB\right \rangle
			&=\frac{1}{8}\sum_{i,j=1}^{3}t_{ij0}a_{i}b_{j}{\rm Tr}(\sigma _{i}^{2}\otimes\sigma _{j}^{2}\otimes I^{2}) \nonumber\\
			&=\sum_{i,j=1}^{3}t_{ij0}a_{i}b_{j} 
			=\left \langle \vec{a},T_{0}\vec{b} \right \rangle
		\end{align}
		where  $\left ( T_{0} \right )_{ij}=t_{ij0}$ and $T=(T_{1},T_{2},T_{3})$, $\left ( T_{k} \right ) _{ij}=t_{ijk}, k=1,2,3$ ,
		\begin{gather*}T_{k}=\begin{pmatrix}
				t_{11k} &t_{12k}  &t_{13k} \\ 
				t_{21k} &t_{22k}  &t_{23k} \\ 
				t_{31k} &t_{32k}  &t_{33k} \end{pmatrix}.
		\end{gather*}
		
		For any given unit vectors $\vec{b}$ and $\vec{b^{'}}$, there always exist a pair of unit vectors $\vec{m}$ and $\vec{m^{'}}$ that are mutually orthogonal, along with an angle $\theta_{b}\in[0,\frac{\pi }{2}]$ such that
		\begin{gather}\label{eq5}
			\vec{b}+\vec{b^{'}}=2\cos\theta_{b}\vec{m}, \ \ \vec{b}-\vec{b^{'}}=2\sin\theta_{b}\vec{m^{'}},
		\end{gather}
		with the measurement settings defined as  $M_{0}=\vec{m}\cdot\vec{\sigma}$, and $M_{1}=\vec{m^{'}}\cdot\vec{\sigma}$, the expectation value $\mathcal{B}_{\rho }$ for any tripartite quantum state can be expressed as
		\begin{align} 
			\mathcal{B}_{\rho }
			&=\left \langle A_{1}B_{-}C_{0} \right \rangle+\left \langle A_{0}B_{+} \right \rangle \nonumber\\
			&=\frac{1}{2}\left \langle A_{1}(B_{0}-B_{1})C_{0}\right \rangle
			+\frac{1}{2}\left \langle A_{0}(B_{0}+B_{1}) \right \rangle  \nonumber\\
			&{=\left \langle A_{1}M_{1}C_{0}\right \rangle\sin\theta_{b}
			+\left \langle A_{0}M_{0} \right \rangle \cos\theta_{b}, \label{eq6} }
		\end{align}
		by substituting ($\ref{eq4}$) into the above expression, we can conclude from the lemma that
		\begin{align}
			\mathcal{B} _{\rho }
			&\leqslant \lambda_{1} \left \| \vec{a^{'}} \right \|\left \| \vec{m^{'}} \right \| \sin\theta_{b} +\lambda_{2} \left \| \vec{a} \right \|\left \| \vec{m} \right \| \cos\theta_{b} \label{eq7}\\
			&\leqslant\sqrt{\lambda_{1}^{2}+\lambda_{2}^{2}}\label{eq8},
		\end{align}
		which proofs the theorem.
	\end{proof}
	
	Now, let's analyze the conditions that ($\ref{eq3}$) must satisfy to achieve a tight upper bound. For the first inequality ($\ref{eq7}$), by appropriately choosing the measurement directions of $A_{i}$, $B_{j}$ and $C_{0}$ can make it become an equality. According to the lemma, if $\vec{a^{'}},\vec{m^{'}}$ and $\vec{a},\vec{m}$ are the singular vectors corresponding to singular values $\lambda_{1}$ and $\lambda_{2}$, respectively, then by applying Cauchy-Schwarz inequality, inequality ($\ref{eq8}$) is saturated, which means that the conditions for equality are satisfied.
	
	In the following, we will conduct an analysis of the upper bound's tightness and randomness through an example presented below.
	\begin{example}
		We consider a bipartite entangled state of mixed	white noise, which is given by
		\begin{gather*}
			\rho = p|\psi \rangle\langle \psi |+\frac{1-p}{8}I,
		\end{gather*}	
		where $I$ is identity matrix on $\mathbb{C}^{2} \otimes \mathbb{C}^{2}\otimes \mathbb{C}^{2}$, $|\psi\rangle=\frac{1}{\sqrt{2}}(|000\rangle+|110\rangle)$, and $0{\leqslant} p{\leqslant} 1$.
		
		The matrix $\sum_{k=1}^{3}c_{k}T_{k}$, $T_{0}$ take the form of
		\begin{gather*}
			\sum_{k=1}^{3}c_{k}T_{k}=
			\begin{pmatrix}
				0  &0  &\ \ \ c_{1}p \\
				0  &0  &-c_{2}p \\
				0  &0  &\ \ \ c_{3}p
			\end{pmatrix}, \\
			T_{0}=			
			\begin{pmatrix}
				p  &0  &0 \\
				0  &-p  &0 \\
				0  &0  &p
			\end{pmatrix}.
		\end{gather*} 
		Thus, the maximum singular values of the matrix $\sum_{k=1}^{3}c_{k}T_{k}$ and $T_{0}$ are
		\begin{gather*}
			\lambda_{\max}\left( \sum_{k=1}^{3}c_{k}T_{k}\right) =\sqrt{(c_{1}^{2}+c_{2}^{2}+c_{3}^{2})p^{2}}=p, \ \
			\lambda _{\max}\left( T_{0}\right) =p.
		\end{gather*}
		Therefore, 
		\begin{align*}
			\mathcal{B}_{Q} 
			&=\underset{A_{x},B_{y},C_{0}}{\max} \left | \mathcal{B}_{\rho} \right |
			=\sqrt{\lambda_{1}^{2}+\lambda_{2}^{2}}
			=\sqrt{2}p.
		\end{align*} 
		
		There exists a unique strategy as shown in ($\ref{eq10}$), which can attain the tight of upper bound of the maximum quantum value, specifically $\sqrt{2}p$. Next, we will  investigate the randomness generated by this strategy.
		To achieve this, it is essential to fulfill  the conditions $\left \langle \vec{m},\vec{m^{'}} \right \rangle =0 $ and $\theta_{b}=\frac{\pi}{4}$, the singular vectors associated with $\sum_{k=1}^{3}c_{k}T_{k}$ and $T_{0}$ can be selected as
		\begin{gather}\label{eq9}
			\vec{a}=\vec{m}=(1,0,0)^{T} ,\ \
			\vec{a^{'}}=\vec{m^{'}}=(0,0,1)^{T}.
		\end{gather}
		By referencing ($\ref{eq5}$) and ($\ref{eq9}$), we can obtain
		\begin{gather*}
			\vec{b}+\vec{b^{'}}={2\cos\theta_{b}\vec{m}}=\left(\sqrt{2},0,0\right)^{T},
			\\
			\vec{b}-\vec{b^{'}}={2\sin\theta_{b}\vec{m^{'}}}=\left(0,0,\sqrt{2} \right) ^{T}.
		\end{gather*}
		The measurement settings are as follows
		\begin{gather}
			\vec{a}=(1,0,0)^{T},\ \
			\vec{a^{'}}=(0,0,1)^{T},  \nonumber\\
			\vec{b}=\frac{\sqrt{2}}{2}\left(1,0,1\right) ^{T}, \
			\vec{b^{'}}=\frac{\sqrt{2}}{2}\left(1,0,-1\right) ^{T},  \nonumber\\
			\vec{c}=(c_{1},c_{2},c_{3})^{T}.\label{eq10}
		\end{gather}
		Without sacrificing generality, we choose the following measurements for each party, 
		\begin{gather*}
			A_{0}=\sigma_{1},\ A_{1}=\sigma_{3}, \\
			B_{0}=\frac{\sqrt{2}}{2}\left( \sigma_{1}+\sigma_{3}\right), \
			B_{1}=\frac{\sqrt{2}}{2}\left( \sigma_{1}-\sigma_{3}\right),\\
			C_{0}=c_{1}\sigma_{1}+c_{2}\sigma_{2}+c_{3}\sigma_{3} .
		\end{gather*}
		The parity-CHSH inequality self-tests the measurements for the state $|\psi\rangle$ 
		which achieves its maximum quantum violation.
		All $A_{x},B_{y},C_{0}$ with eigenvalues $\pm1$, the corresponding projection measurements for eigenvalue can be expressed as $M_{a|x}=\frac{I\pm A_{x}}{2}$, $M_{b|y}=\frac{I\pm B_{y}}{2}$, $M_{c|0}=\frac{I\pm C_{0}}{2}$, 		
		\begin{gather}
			M_{1|0}^{A} =\frac{1}{2}\begin{pmatrix}
				1  &1 \\
				1  &1
			\end{pmatrix},\
			M_{-1|0}^{A} =\frac{1}{2}\begin{pmatrix}
				1  &-1 \\
				-1  &1
			\end{pmatrix}, \nonumber\\
			M_{1|1}^{A} =\frac{1}{2}\begin{pmatrix}
				1  &0 \\
				0  &0
			\end{pmatrix},\
			M_{-1|1}^{A} =\frac{1}{2}\begin{pmatrix}
				0  &0 \\
				0  &1
			\end{pmatrix}, \nonumber\\
			M_{1|0}^{B} =\frac{\sqrt{2}}{4}\begin{pmatrix}
				\sqrt{2}+1  &1 \\
				1  &\sqrt{2}-1
			\end{pmatrix},  \nonumber\\
			M_{-1|0}^{B} =\frac{\sqrt{2}}{4}\begin{pmatrix}
				\sqrt{2}-1  &-1 \\
				-1  &\sqrt{2}+1
			\end{pmatrix}, \nonumber\\
			M_{1|1}^{B} =\frac{\sqrt{2}}{4}\begin{pmatrix}
				\sqrt{2}-1  &1 \\
				1  &\sqrt{2}+1
			\end{pmatrix},  \nonumber\\
			M_{-1|1}^{B} =\frac{\sqrt{2}}{4}\begin{pmatrix}
				\sqrt{2}+1  &-1 \\
				-1  &\sqrt{2}-1
			\end{pmatrix}, \nonumber\\
			M_{1|0}^{C} =\frac{1}{2}\begin{pmatrix}
				1+c_{3}  &c_{1}-c_{2}i \\
				c_{1}+c_{2}i  &1-c_{3} 
			\end{pmatrix},\nonumber\\
			M_{-1|0}^{C} =\frac{1}{2}\begin{pmatrix}
				1-c_{3}  &-c_{1}+c_{2}i \\
				-c_{1}-c_{2}i  &1+c_{3} 
			\end{pmatrix}. \label{eq11}
		\end{gather}
		
		The corresponding quantum correlation is generated by the Born rule: $P(a,b,c|x,y,z)={\rm Tr}\left(\rho M_{a|x}\otimes M_{b|y}\otimes M_{c|z} \right) $, $a,b,c\in \left\lbrace 1,-1\right\rbrace$, $x,y\in \left\lbrace 0,1\right\rbrace, z=0 $. This leads to a total of 32 joint probabilities. Due to the selection of measurement operators, there may be instances where these probabilities are equal. We can classify these equal probabilities into categories for comparison and identify the largest one among them,
		\begin{gather*}
			P_{1}=\frac{(2-\sqrt{2})c_{3}-\sqrt{2} }{16}p+\frac{1}{8},\\
			P_{2}=\frac{-(2+\sqrt{2})c_{3}+\sqrt{2} }{16}p+\frac{1}{8}, \\
			P_{3}=\frac{(-2+\sqrt{2})c_{3}-\sqrt{2} }{16}p+\frac{1}{8}, \\
			P_{4}=\frac{(2+\sqrt{2})c_{3}+\sqrt{2} }{16}p+\frac{1}{8}.
		\end{gather*}
		It is evident that $P_{1},P_{2},P_{3} {\leqslant} P_{4}$. Hence, regarding parity-CHSH inequality, the maximum probability achievable among all probability distributions is
		\begin{gather}\label{eq12}
			\underset{\left \{ M_{a|x},M_{b|y},M_{c|z} \right \} }{\max} P(abc|xyz)=\frac{1+\sqrt{2}}{8}p+\frac{1}{8}.
		\end{gather}
	\end{example}
	
	\section{RANDOMNESS BASED ON THE PARITY-CHSH INEQUALITY}
	The violation of Bell inequality provides evidence for the existence of quantum entanglement and, in a DI environment, demonstrates the existence of randomness. Owing to the uncertainty principle of quantum mechanics, quantum systems possess inherent randomness that can be harnessed to generate random numbers and violate Bell inequality. Entropy quantifies the randomness within a system and is pivotal in applications such as random number generation.
	
	In this work, we employ the min-entropy
	\begin{align*}
		H_{\min}=-\log_{2}{\max_{abc}P(a,b,c|x,y,z)},
	\end{align*}	
	to quantify the randomness of output pairs $(a,b,c)$ limited by the input pairs $(x,y,z)$.
	For a given violation value I of the Bell inequality, the problem of solving the maximum probability can be reformulated as the following optimization problem:
	\begin{align}\label{eq13}
		\max_{a,b,c} \quad 
		& P(abc|xyz) \nonumber\\
		\mbox{s.t.} \quad
		&\sum_{a,b,c,x,y,z}\alpha_{abc|xyz}P(abc|xyz)=I \nonumber\\
		&P(abc|xyz)={\rm Tr}\left(\rho_{abc}M_{a|x}\otimes M_{b|y}\otimes M_{c|z} \right).
	\end{align}	
	However, the carving of quantum correlation sets is challenging. Obtaining $\max_{abc}P(a,b,c|x,y,z)$ requires substantial computational effort and necessitates an exhaustive search through all $\left\lbrace \rho,M_{a|x},M_{b|y},M_{c|z} \right\rbrace $. 
	The proposal of NPA hierarchy has become an important tool for studying DI protocols.   
	This technology allows us to characterize the set of quantum correlations by using a series of tight external approximations, each of which can be represented algorithmically as a set of feasible solutions to a SDP problem.
	
	As shown in FIG. $\ref{Figure1}$, a numerical search for the solution to the optimization problem yields an upper bound on the maximum probability, also referred to as the lower bound of the min-entropy at the $Q_{aq}$ level of the NPA hierarchy, as discussed in Ref. \cite{NPA3}. 
	Under the condition of the tight upper bound of the maximum violation value of the parity-CHSH inequality, we derived a lower bound of the maximum probability ($\ref{eq12}$) based on a bipartite entangled state within tripartite system, which acts as an upper bound on randomness.
	When p approaches 1, the amount of
	random bit can reach up to 1.2284. This result is consistent with the randomness obtained through the NPA algorithm, which suggests that we have found the exact solution for the minimum randomness.
	Furthermore, we observe when the noise level satisfies  $0.7071\approx1/\sqrt{2}<p {\leqslant} 1$, randomness can be extracted.
	
	\begin{figure}[H]
		\centering
		\includegraphics[width=0.4\textwidth]{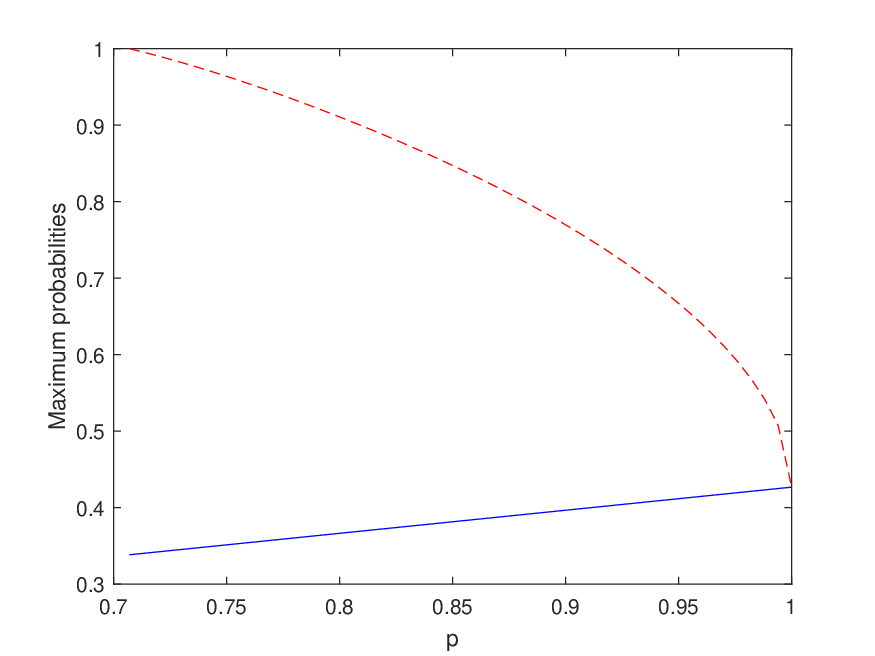}
		\caption{The maximum probabilities $P(abc|xyz)$ vs visibility p, optimally violates the parity-CHSH inequality in the presence of white noise. The blue curve is the lower bound of the maximum probability obtained under the measurement described in ($\ref{eq12}$). The red dashed curve represents the upper bound of the maximum probability obtained at the $Q_{aq}$ level of NPA hierarchy. When p=1, which signifies the absence of noise effects, the lower and upper bounds for the maximum probability are both 0.4268. } \label{Figure1}
	\end{figure}

	\section{monogamy relations for four-qubit GHZ-class states}
	We will now consider the monogamy relation of the parity-CHSH inequality for the four-qubit the GHZ-class states $|\psi_{gs}\rangle$. The corresponding density matrices are defined as $\rho_{abcd}=|\psi_{gs}\rangle\langle \psi_{gs}|$,
	\begin{gather*}
		|\psi_{gs}\rangle=\cos\theta|0000\rangle+\sin\theta|111\rangle\left(\cos\varPhi |0\rangle+\sin\varPhi |1\rangle \right) .
	\end{gather*}	
	The reduced three-qubit density matrices derived from $\rho_{abcd}$ are labeled as $\rho_{abc}={\rm Tr}_{d}\left(\rho_{abcd}\right) $, $\rho_{abd}={\rm Tr}_{c}\left(\rho_{abcd}\right) $, $\rho_{acd}={\rm Tr}_{b}\left(\rho_{abcd}\right) $ and $\rho_{bcd}={\rm Tr}_{a}\left(\rho_{abcd}\right) $.
	\begin{widetext}
		\begin{theorem}
			\label{thm2}
			For the four-qubit GHZ-class states $\rho_{abcd}=|\psi_{gs}\rangle\langle \psi_{gs}|$, the monogamy relation between  the quantum value of the parity-CHSH inequality for reduced tripartite states are satisfied as follows:
			\begin{align} \label{eq14}
				\mathcal{B}^{2} _{\rho_{abc}}+\mathcal{B}^{2} _{\rho_{abd}}+\mathcal{B}^{2} _{\rho_{acd}}+\mathcal{B}^{2} _{\rho_{bcd}}
				\leqslant 4 .
			\end{align}
			where,
			\begin{align*}
				\rho_{abc}
				&=\cos^{2}\theta |000\rangle\langle 000| +\sin^{2}\theta |111\rangle\langle 111|  +\cos\theta\sin\theta\cos\varPhi\left(|000\rangle\langle 111|+|111\rangle\langle 000|\right), \\
				\rho_{abd} &=\rho_{acd}=\rho_{bcd} \\
				&=\cos^{2}\theta |000\rangle\langle 000| +\sin^{2}\theta\sin^{2}\varPhi |111\rangle\langle 111| \\
				&\quad+\sin^{2}\theta\cos^{2}\varPhi |110\rangle\langle 110|+\sin^{2}\theta\sin\varPhi\cos\varPhi \left( |110\rangle\langle 111|+|111\rangle\langle 110| \right) .
			\end{align*}	
		\end{theorem}
		\begin{proof}[proof]
			Noting $\vec{x}=\left( \sin\theta_{x}\cos\varPhi_{x}, \sin\theta_{x}\sin\varPhi_{x}, \cos\theta_{x} \right) $, $\vec{x}\in \left \{ a,a^{'} ,b,b^{'},c,c^{'},d,d^{'}\right \}$, we find that, according to Eq. ($\ref{eq5}$), for $\vec{y}\in \left \{ b,c\right \}$, $\vec{m_{y}}$ and $\vec{m_{y}^{'}}$ satisfy orthogonality. This imply that
			\begin{align*} 
				\vec{m_{y}}\cdot\vec{m_{y}^{'}}
				&=\cos\theta_{m_{y}}\cos\theta_{m_{y}^{'}} +\sin\theta_{m_{y}}\sin\theta_{m_{y}^{'}}\cos\left( \varPhi_{m_{y}}-\varPhi_{m_{y}^{'}} \right) =0,
			\end{align*}
			by further derivation we find that $\cos^{2}\theta_{m_{y}}+\cos^{2}\theta_{m_{y}^{'}} {\leqslant} 1$, while $\sin^{2}\theta_{m_{y}}+\sin^{2}\theta_{m_{y}^{'}} {\leqslant} 2$. 
			On the basis of Eq. ($\ref{eq6}$), for the reduced state $\rho_{abc}$ 
			\begin{align} \label{eq15}
				\mathcal{B}^{2} _{\rho_{abc}} 
				&\leqslant \left \langle A_{1}M_{b_{1}}C_{0}\right\rangle^{2}
				+\left \langle A_{0}M_{b_{0}}\right \rangle^{2}.
			\end{align} 
			It is easy to calculate $\left \langle A_{1}M_{b_{1}}C_{0}\right \rangle^{2}$ and $\left \langle A_{0}M_{b_{0}}\right \rangle^{2}$,
			
			\begin{align*} 
				\left \langle A_{0}M_{b_{0}}\right \rangle^{2}_{\rho_{abc}}
				&=\cos^{2}\theta_{a}\cos^{2}\theta_{m_{b}}, \\
				\left \langle A_{1}M_{b_{1}}C_{0}\right \rangle^{2}_{\rho_{abc}	}
				&= \left[\cos\theta_{a^{'}}\cos\theta_{c}\cos\theta_{m^{'}_{b}}\cos2\theta
				+\sin\theta_{a^{'}}\sin\theta_{c}\sin\theta_{m^{'}_{b}}\cos \left(\varPhi_{a^{'}}+\varPhi_{c}+\varPhi_{m^{'}_{b}}\right) \sin2\theta\cos\varPhi    \right] ^{2}  .
			\end{align*}	
			Therefore, we obtian
			\begin{align*} 
				\mathcal{B}^{2} _{\rho_{abc}} 
				&\leqslant \cos^{2}\theta_{a}\cos^{2}\theta_{m_{b}}
				+\left[ \cos\theta_{a^{'}}\cos\theta_{c}\cos\theta_{m^{'}_{b}}\cos2\theta
				+\sin\theta_{a^{'}}\sin\theta_{c}\sin\theta_{m^{'}_{b}} \sin2\theta\cos\varPhi    \right] ^{2} \nonumber\\
				&\leqslant\cos^{2}\theta_{m_{b}}
				+\cos^{2}\theta_{a^{'}}\cos^{2}\theta_{m^{'}_{b}}\cos^{2}2\theta
				+\sin^{2}\theta_{a^{'}}\sin^{2}\theta_{m^{'}_{b}} \sin^{2}2\theta\cos^{2}\varPhi  .
			\end{align*} 
			
			For the reduced states $\rho_{abd}$, $\rho_{acd}$ and $\rho_{bcd}$, similar calculations can be performed as follows:
			\begin{align*} 
				\mathcal{B}^{2} _{\rho_{abd}}
				&\leqslant  \left \langle A_{0}M_{b_{0}}\right \rangle^{2}
				+\left \langle A_{1}M_{b_{1}}D_{0}\right\rangle^{2}   \\
				&=\cos^{2}\theta_{a}\cos^{2}\theta_{m_{b}}
				+\cos^{2}\theta_{a^{'}}\cos^{2}\theta_{m^{'}_{b}}
				\left[ \sin\theta_{d^{'}}\cos\varPhi_{d^{'}}\sin^{2}\theta\sin 2\varPhi+\cos\theta_{d^{'}}\left(1-2\sin^{2}\theta\sin^{2}\varPhi \right)  \right] ^{2}  \\
				&\leqslant  \cos^{2}\theta_{m_{b}}
				+\cos^{2}\theta_{a^{'}}\cos^{2}\theta_{m^{'}_{b}} \left( 1-\sin^{2}2\theta\sin^{2}\varPhi \right),\\
				\mathcal{B}^{2} _{\rho_{acd}}
				&\leqslant  \cos^{2}\theta_{m_{c}}
				+\cos^{2}\theta_{a^{'}}\cos^{2}\theta_{m^{'}_{c}} \left( 1-\sin^{2}2\theta\sin^{2}\varPhi \right) , \\
				\mathcal{B}^{2} _{\rho_{bcd}}
				&\leqslant  \cos^{2}\theta_{m_{c}}
				+\cos^{2}\theta_{b^{'}}\cos^{2}\theta_{m^{'}_{c}} \left( 1-\sin^{2}2\theta\sin^{2}\varPhi \right) 
				\leqslant  \cos^{2}\theta_{m_{c}}
				+\cos^{2}\theta_{m^{'}_{c}} \left( 1-\sin^{2}2\theta\sin^{2}\varPhi \right)		 .
			\end{align*}	
			As a result, we can obtain
			\begin{align*}
				\sum_{a\leqslant x<y<z\leqslant d}\mathcal{B}^{2} _{\rho_{xyz}}=
				&\mathcal{B}^{2} _{\rho_{abc}}+\mathcal{B}^{2} _{\rho_{abd}}+\mathcal{B}^{2} _{\rho_{acd}}+\mathcal{B}^{2} _{\rho_{bcd}} \\
				&\leqslant
				2\cos^{2}\theta_{m_{b}}
				+\sin^{2}\theta_{m^{'}_{b}}\sin^{2}\theta_{a^{'}}\sin^{2}2\theta\cos^{2}\varPhi  
				+\cos^{2}\theta_{m^{'}_{b}}\cos^{2}\theta_{a^{'}} \left(\cos^{2}2\theta+1-\sin^{2}2\theta\sin^{2}\varPhi\right) \\
				&\quad+2\cos^{2}\theta_{m_{c}}
				+\cos^{2}\theta_{m^{'}_{c}}\cos^{2}\theta_{a^{'}}\left(1-\sin^{2}2\theta\sin^{2}\varPhi\right)
				+\cos^{2}\theta_{m^{'}_{c}}\left(1-\sin^{2}2\theta\sin^{2}\varPhi\right) \\
				&\leqslant	
				2\left(1-\cos^{2}\theta_{m^{'}_{b}}\right) 		
				+\left( 1-\cos^{2}\theta_{m^{'}_{b}} \right)\sin^{2}\theta_{a^{'}}\sin^{2}2\theta\cos^{2}\varPhi  
				+\cos^{2}\theta_{m^{'}_{b}}\cos^{2}\theta_{a^{'}} \left(\cos^{2}2\theta+1-\sin^{2}2\theta\sin^{2}\varPhi\right) \\
				&\quad +2\left(1-\cos^{2}\theta_{m^{'}_{c}}\right)
				+\cos^{2}\theta_{m^{'}_{c}}\cos^{2}\theta_{a^{'}}\left(1-\sin^{2}2\theta\sin^{2}\varPhi\right) +\cos^{2}\theta_{m^{'}_{c}}\left(1-\sin^{2}2\theta\sin^{2}\varPhi\right)\\
				&=4+\cos^{2}\theta_{m^{'}_{b}}\left[ 2\cos^{2}\theta_{a^{'}}\left( 1-\sin^{2}2\theta\sin^{2}\varPhi \right)-2-\sin^{2}2\theta\cos^{2}\varPhi  \right]
				-\cos^{2}\theta_{a^{'}}\sin^{2}2\theta\cos^{2}\varPhi \\
				&\quad
				+\cos^{2}\theta_{m^{'}_{c}}\left[ \left(1-\sin^{2}2\theta\sin^{2}\varPhi\right)\left(1+\cos^{2}\theta_{a^{'}}\right)-2  \right] \\
				&\leqslant	4.
			\end{align*}
			
			The monogamy relation presented in Theorem $\ref{thm2}$ has been demonstrated.
		\end{proof}
	\end{widetext}
	The equality holds when 
	\begin{gather*}
		\theta_{a}=\theta_{b}=\theta_{b^{'}}=\theta_{c}=0, \ \ 
		\varPhi_{a^{'}}+\varPhi_{c}+\varPhi_{m_{b^{'}}}=0, \ \
		\varPhi_{d^{'}}=0,\\
		\theta_{a^{'}}=\theta_{m^{'}_{b}}=\theta_{m^{'}_{c}}=\frac{\pi}{2}, \ \ 
		\cos\theta_{d^{'}}=\frac{1-2\sin^{2}\theta\sin^{2}\varPhi}{\sqrt{1-\sin^{2}2\theta\sin^{2}\varPhi}} .
	\end{gather*}
	
	In general, the maximum quantum value of the parity-CHSH inequality is $\mathcal{B}_{Q}=\sqrt{2}$, suggesting that $\mathcal{B}^{2} _{\rho_{xyz}}\leqslant 2$. However, the right side of Eq. ($\ref{eq14}$) equals 4 rather than 8, indicating that the monogamy relation ($\ref{eq14}$) introduces constraints on how nonlocality is distributed among the subsystems. Consequently, $\mathcal{B}^{2} _{\rho_{xyz}}=2$ cannot be satisfied for any tripartite GHZ-class states $\rho_{xyz}$. 
	This further means it is impossible to simultaneously violate the parity-CHSH inequality. In summary, for all combinations $xyz=abc,abd,acd,bcd$ $\mathcal{B}^{2} _{\rho_{xyz}}>1$ will not hold at the same time.

	\section{CONCLUSION AND DISCUSSION}
	The nonlocality of Bell inequality plays a significant role and meaning in various fields. Therefore, this work investigates its randomness by introducing a method for deriving the maximum violation value of the parity-CHSH inequality for any tripartite quantum state. Our analysis reveals that the upper bound of the maximum violation value is tight, and corresponding constraints are given. Specifically, for bipartite entangled state $|\psi\rangle$ with white noise, we calculate the maximum quantum value of the tripartite parity-CHSH inequality, and obtain the necessary and sufficient conditions for its tight upper bound through self-testing of measurements.
	
	Randomness, as a fundamental characteristic of nature, is the key factor behind DI randomness generation (DIRG). Therefore, we study its randomness and maximum probability by considering violations of the specified Bell inequality. 
	When the parity-CHSH inequality achieves the tight upper bound of the maximum violation value, we derive the corresponding maximum probability, which can be used as an upper bound on the min-entropy for quantifying randomness.
	Without regard to the presence of noise, we obtain the lower bound of DI randomness through the SDP method. The consistency between the upper and lower bounds implies this method allows us to achieve the precise solution for the minimum randomness.
	However, for the maximally entangled state, although it can achieve the maximum quantum value, it cannot reach the lower bound of the maximum probability. Besides this, We also investigate the monogamy relation of genuine tripartite nonlocality within four-qubit systems and a tight upper bound for the corresponding GHZ-class states is given based on the maximum quantum value of the parity-CHSH operator. Our result reveals that in a four-qubit quantum system, tripartite nonlocality is interdependent rather than mutually independent.
	
	Constructing tight upper bounds of Bell inequalities for arbitrary quantum states to certify corresponding randomness constitutes a crucial research direction in foundational studies of quantum nonlocality and randomness generation.
	It is worth further studying whether this method is effective when extended to Bell scenarios with multiple measurement settings or more parties.
	Regarding the monogamy relations of parity-CHSH inequalities, the fundamental question of whether analogous constraints emerge for arbitrary quantum states presents significant theoretical implications. This intriguing problem still needs to be systematically explored, and we will leave it for future.

	\section*{\bf acknowledgments} This work is supported by the Shandong Provincial Natural Science Foundation for Quantum Science ZR2021LLZ002, the Fundamental Research Funds for the Central Universities No.22CX03005A.
	
	\section*{\bf Data availability statement} All data generated or analyzed during this study are included in this published article.

\end{document}